\documentclass[times, 11pt]{article}
\usepackage{times}
\usepackage{epsfig}
\usepackage{epstopdf}
\usepackage[margin=1in]{geometry}
\usepackage {amssymb,latexsym,amsthm,amsmath,enumerate,amscd,epsfig}
\usepackage{setspace}
\newtheorem{theorem}{Theorem}

\newtheorem{example}{Example}
\newtheorem{definition}{Definition}
\newtheorem{lemma}{Lemma}
\newtheorem{corollary}{Corollary}
\newtheorem{proposition}{Proposition}

\begin{document}

\title{Mixed Covering Arrays on 3-Uniform Hypergraphs }

\author{Yasmeen Akhtar$~~~~$ Soumen Maity\\
Indian Institute of Science Education and Research\\
Pune, India\\}
\date{}
\maketitle

\begin{center}\textbf{\large Abstract}
\end{center}{\it Covering arrays are combinatorial objects that have been successfully applied in the design of test suites for testing systems such as software, circuits and networks,
where failures can be caused by the interaction between their parameters. In this paper, we perform a new generalization of covering arrays called covering arrays on 3-uniform hypergraphs.
Let $n, k$ be positive integers with $k\geq 3$. Three vectors $x\in \mathbb Z_{g_1}^n$,  $y\in \mathbb Z_{g_2}^n$, $z\in \mathbb Z_{g_3}^n$ are {\it 3-qualitatively independent} if for any triplet $(a, b, c) \in  \mathbb Z_{g_1}\,\times\, \mathbb Z_{g_2}\,\times\,\mathbb Z_{g_3}$,
there exists an index $ j\in  \lbrace 1, 2,...,n \rbrace $ such that $( x(j), y(j), z(j)) = (a, b, c)$. Let $H$ be a 3-uniform hypergraph with $k$ vertices $v_1,v_2,\ldots,v_k$ with 
respective vertex weights $g_1,g_2,\ldots,g_k$. A mixed covering array on  $H$, denoted by  $3-CA(n,H, \prod_{i=1}^{k}g_{i})$, is a $k\times n$ 
array such that row $i$ corresponds to vertex $v_i$, entries in row $i$ are from $Z_{g_i}$; and if $\{v_x,v_y,v_z\}$ is a hyperedge in $H$, then the rows $x,y,z$ are 3-qualitatively 
independent.  The parameter $n$ is called the size of the array. Given a weighted  3-uniform hypergraph $H$, a mixed covering array on $H$ with minimum size is called optimal. 
We outline necessary background in the theory of hypergraphs  that is relevant to the study of covering arrays on
hypergraphs. In this article, we introduce five basic hypergraph operations to construct optimal mixed covering arrays on hypergraphs. Using these operations,  we provide constructions for optimal mixed covering arrays on $\alpha$-acyclic 3-uniform hypergraphs, conformal 3-uniform hypertrees having a binary tree as host tree,  and on some 
specific 3-uniform cycle hypergraphs. }\\


\section{Introduction}
Covering arrays have been extensively studied and have been the topic of interest of many researchers.
These interesting mathematical structures are generalizations of well known orthogonal arrays \cite{heday}.
 A covering array of strength three, denoted by 3-$CA( n, k, g )$, is an $k  \times  n $ array $C$ with entries from $\mathbb Z_{g}$
such that any three distinct rows of $ C $ are 3-qualitatively independent. The parameter $n$ is called the size of the array. {\it One of the main problems on covering arrays is to construct a
3-$CA(n,k,g)$ for given parameters $(k,g)$ so that the size $n$ is as small as possible. }  The covering array number 3-$CAN(k,g)$ is the smallest $n$ for which a 3-$CA(n,k,g)$ exists, that is
\begin{center}
3-$CAN( k, g ) = min_{n\in\mathbb N}\,\lbrace  n ~|~ \exists $ 3-$CA(n,k,g)\,\rbrace$.
\end{center}
 A 3-$CA( n, k, g ) $ of
size $ n  = $ 3-$CAN( k, g )$ is called $ optimal $.
An example of a strength three covering array 3-$CA( 10, 5, 2 )$ is shown below \cite{kreher}:
\begin{center}
1\hspace{3mm}0\hspace{3mm}1\hspace{3mm}0\hspace{3mm}1\hspace{3mm}0\hspace{3mm}0\hspace{3mm}0\hspace{3mm}1\hspace{3mm}1\\
1\hspace{3mm}0\hspace{3mm}1\hspace{3mm}0\hspace{3mm}0\hspace{3mm}1\hspace{3mm}0\hspace{3mm}1\hspace{3mm}0\hspace{3mm}1\\
1\hspace{3mm}0\hspace{3mm}0\hspace{3mm}1\hspace{3mm}0\hspace{3mm}0\hspace{3mm}1\hspace{3mm}1\hspace{3mm}1\hspace{3mm}0\\
1\hspace{3mm}0\hspace{3mm}0\hspace{3mm}1\hspace{3mm}1\hspace{3mm}1\hspace{3mm}0\hspace{3mm}0\hspace{3mm}0\hspace{3mm}1\\
1\hspace{3mm}0\hspace{3mm}1\hspace{3mm}0\hspace{3mm}1\hspace{3mm}1\hspace{3mm}1\hspace{3mm}0\hspace{3mm}0\hspace{3mm}0\\
\end{center}
\noindent There is a vast array of literature \cite{Hartman,Col1,Colbourn,chatea,kreher,karen} on covering arrays, and the problem of determining the minimum size of covering arrays has been studied under many guises over the past thirty years.

Covering arrays have applications in many areas. Covering arrays are particularly useful in the design of test suites \cite{Hartman,Cohen,patton,maity,maity1, maity2}.
The testing application is based on the following translation. Consider a software system that has $k$ parameters, each parameter can take $g$ values. Exhaustive testing would require $g^k$ test cases for detecting software failure, but if $k$ or $g$  are reasonably large, this may be infeasible. We wish to build a test suite that tests all 3-way interactions of parameters with the minimum number  of test cases.
 Covering arrays of strength 3 provide compact test suites that guarantee 3-way coverage of parameters.

Several generalizations of covering arrays have been proposed in order to address different requirements of the testing application (see \cite{Col1,HartmanDM} ).
{\it Mixed covering arrays} are a generalization of covering arrays that allows different values for different rows. This meets the requirement that different parameters in the system
may take a different number of possible values. Constructions for mixed covering arrays are given in \cite{c.j,moura}. Another generalization of
covering arrays  are {\it mixed covering arrays on hypergraph}. In these arrays, only specified choices of distinct rows need to be
qualitatively independent and these
choices are recored in hypergraph. As mentioned in \cite{mixed}, this is useful in situations in which some combinations of parameters do not interact;
in these cases, we do not insist that these interactions to be tested, which allows reductions in the number of required test cases.
This has been applied in the context of software testing by observing that we only need to test interactions between parameters that jointly effect one of the output values \cite{cheng}.
Covering arrays on graphs were first studied by Serroussi and Bshouty \cite{ser}, who showed that finding an optimal covering array on a graph is NP-hard for the binary case.
Covering arrays on general alphabets have been systematically studied in Steven's thesis \cite{ste}.
Meagher and
Stevens \cite{karen}, and Meagher, Moura, and Zekaoui \cite{mixed} studied  strength two (mixed) covering arrays on graphs in more details and gave many powerful results.
Variable strength covering arrays have been introduced and systematically studied  in Raaphorst's thesis \cite{sebastian}.

\par In this paper, we extend the work done by  Meagher, Moura, and Zekaoui \cite{mixed} for mixed covering arrays on graph to mixed covering arrays on hypergarphs. The motivation 
for this generalisation is to improve applications of covering arrays to software, circuit and network systems. This extension also gives us new ways to study covering arrays construction.  In Section 2, we outline necessary background in the theory of hypergraphs and mixed covering arrays that are relevant to the study of mixed covering arrays on
hypergraphs. In Section 3,  we present  results related to  balanced and pairwise balanced vectors which are required for  basic hypergraph operations. In section 4,    we introduce four basic hypergraph operations. Using these operations, we construct optimal mixed covering arrays on $\alpha$-acyclic 3-uniform hypergraphs, conformal 3-uniform 
hypertrees having a binary tree as host tree, some specific 3-uniform cycles. In Section 5, we build optimal mixed covering arrays on 3-uniform cycles with exactly one vertex of degree one. 

\section{Mixed covering arrays and  hypergraphs}

A mixed covering array is a generalization of covering array that allows  different alphabets in different rows. 

\begin{definition} (Mixed Covering Array) Let $n,k,g_1, \ldots,g_k$ be positive integers. A mixed covering array of strength three, denoted by $3-CA(n,k, \prod_{i=1}^k{g_i})$ is an $k\times n$ array $C$
with entries from $\mathbb Z_{g_i}$ in row $i$, such that any three distinct rows of $C$ are 3-qualitatively independent.
\end{definition}

\noindent The parameter $n$ is called the size of the array. An obvious lower bound for the size of a covering array is $g_ig_jg_k$ where $g_i$, $g_j$, $g_k$ are the
largest three alphabets, in order to guarantee that the corresponding three rows be 3-qualitatively independent.

\begin{definition} (Hypergraphs \cite{berge})  A hypergraph $H$ is a pair  $H = (V,E)$ where $V=\{v_1,v_2,\ldots,v_k\}$ is a set of elements called nodes or vertices, 
and $E=\{E_1,E_2,\ldots,E_m\}$ is a set of non-empty subsets of $V$,  called hyperedges, such that 
$$E_i\neq \emptyset ~~~~~(i=1,2,\ldots m)$$
$$ \bigcup_{i=1}^m {E_i}=V.$$
A simple hypergraph is a hypergraph $H$ such that $$E_i\subset E_j  \Rightarrow i=j.$$

\end{definition}

\noindent If cardinality of every hyperedge of $H$ is equal to $r$ then $H$ is called $r$-uniform hypergraph. A complete $r$-uniform hypergraph containing $k$ vertices,
denoted by $K_k^r$, is a hypergraph having every $r$-subset of set of vertices as hyperedge. For a set $J \subset\{1,2,...,m\}$, 
the partial hypergraph generated by $J$ is the hypergraph
$\left(V, \lbrace E_i | i\in J \rbrace \right)$.
 For a set $A\subset V$, the subhypergraph $H_A$ induced by $A$  is defined as
$H_A=\left(A, \{E_j \cap A \mid 1\leq i \leq m, E_i \cap A \neq \emptyset\} \right)$.
The 2-section  of a hypergraph $H$  is the graph $[H]_2$ with the same vertices of the hypergraph, and edges between all 
pairs of vertices contained in the same hyperedge. 
\begin{definition}(Conformal Hypergraph \cite{berge})
 A hypergraph $H$ is conformal if all the maximal cliques of the graph $[H]_2$ are hyperedges of $H$.
\end{definition}

\begin{definition} (Tripartite 3-uniform hypergraph \cite{berge})
A tripartite 3-uniform hypergraph is a  3-uniform hypergraph in which the set of vertices is $V_1 \cup V_2\cup V_3$ and the hyperedges are the 3-tuples $ \lbrace v_{1}, v_{2},v_{3}\rbrace $ with $ v_i \in V_{i} $ for $ i = 1,2,3 $.
\end{definition}

\begin{definition}\cite{berge} 
  Let $H$ be a hypergraph on $V$, and let $k\geq 2$ be an integer. A cycle of length $k$ is a sequence 
  $(v_1,E_1,v_2,E_2,...,v_k,E_k,v_1)$ with:
  \begin{enumerate}
   \item $E_1,E_2,...,E_k$ distinct hyperedges of $H$;
   \item $v_1,v_2,...,v_k$ distinct vertices of $H$;
   \item $v_i,v_{i+1}\in E_i$  for $i=1,2,\ldots,k-1$;
   \item $v_k,v_1 \in E_k$.
  \end{enumerate}
\end{definition}

\begin{definition} (Balanced Hypergraphs \cite{berge})
A hypergraph is said to be balanced if every odd cycle has a hyperedge containing three vertices of the cycle.
\end{definition}

\begin{theorem}\label{2color}\cite{berge} A hypergraph is balanced if and only if its induced subhypergraphs are 2-colourable.
\end{theorem}

\noindent  A {\it vertex-weighted hypergraph} is a hypergraph with
a positive weight assigned to each vertex. We give here the definition of mixed covering array on hypergraph: 
\begin{definition} Let $ H $ be a vertex-weighted  hypergraph with $ k $ vertices and weights  $ g_{1} \leq  g_{2} \leq ... \leq g_{k}$, and let $ n $ be a positive integer.
A covering array on $ H $, denoted by $CA( n, H, \prod_{i=1}^{k}g_{i})$, is an $ k \times n $ array with the following properties:
\begin{enumerate}
\item the entries in row $i$ are from  $\mathbb Z_{g_{i}}$;

\item  row $ i $ corresponds to a vertex $ v_i \in V(H) $ with weight $g_{i} $;

\item  if $e=\{v_1,v_2,\ldots,v_t\}\in E(H)$, the rows correspond to vertices $v_1,v_2,\ldots,v_t$ are $t$-qualitatively independent. 
\end{enumerate}
\end{definition}

\noindent In this paper we concentrate on covering arrays on 3-uniform hypergraphs.  Given a weighted 3-uniform hypergraph $ H $ with weights $ g_{1}, g_{2},..., g_{k} $ a {\it strength-3 mixed covering array} on $H$ is denoted by 3-$CA( n, H, \prod_{i=1}^{k}g_{i} )$; {\it the strength-3 mixed covering array number} on $ H $,
denoted by 3-$CAN(H, \prod_{i=1}^{k}g_{i} )$, is the minimum $n$ for which there exists a 3-$CA( n, H, \prod_{i=1}^{k}g_{i} )$.
A 3-$CA( n, H, \prod_{i=1}^{k}g_{i} ) $ of size $ n$ = 3-$CAN( H, \prod_{i=1}^{k}g_{i} )$ is called $ optimal $.
A mixed covering array of strength three, denoted by 3-$CA( n, k, \prod_{i=1}^{k}g_{i} )$, is a 3-$CA( n, K_{k}^{3}, \prod_{i=1}^{k}g_{i} )$, where $ K_{k}^{3} $ is the complete 3-uniform hypergraph on $ k $ vertices with weights $ g_{i} $, for $ 1 \leq i \leq k $.

\section{Balanced and Pairwise Balanced Vectors}

In this section, we present several results related to  balanced and pairwise balanced vectors which are required for basic hypergraph operations defined in the next section. 
\begin{definition}
 A length-$n$ vector with alphabet size $g$ is  $ \textit{balanced} $ if each symbol occurs $ \lfloor n/g \rfloor $
or $ \lceil n/g \rceil $ times.
\end{definition}

\begin{definition}
 Two length-$ n $ vectors $ x_{1} $  and $ x_{2} $ with alphabet size $ g_{1} $ and $ g_{2} $ are $ \textit{pairwise balanced} $ if both vectors are
balanced and each pair of 
alphabets $ (a, b)\in \mathbb Z_{g_{1}} \times \mathbb Z_{g_{2}}$ occurs  $ \lfloor n/g_{1}g_{2} \rfloor $ or 
$ \lceil n/g_{1}g_{2} \rceil $ times in $ (x_{1}, x_{2}) $, so for $n\geq g_1g_2$ pairwise balanced vectors are always 2-qualitatively independent.
\end{definition}

\begin{definition} Let $H$ be a vertex-weighted hypergraph. 
A balanced covering array on  $H$ is 
 a covering array on $H$ in which every row is balanced and the rows correspond to vertices in a hyperedge are pairwise balanced. 
\end{definition}

\begin{lemma}\label{lemma1}
 Let $ x_{1} \in \mathbb Z_{g_{1}}^{n}$ and $ x_{2} \in \mathbb Z_{g_{2}}^{n} $ be two balanced vectors.
 Then for any positive integer $h$, there exists a balanced vector $y\in \mathbb Z_{h}^{n}$ 
such that $x_1$ and $y$ are pairwise balanced and $x_2$ and $y$  are pairwise balanced. 
\end{lemma}

\begin{proof}
 Construct a bipartite multigraph $G$ corresponds to $x_1$ and $x_2$ as follow: $G$ has $g_1$ vertices in the first part 
 $P\subseteq V(G)$ and $g_2$ vertices in the second part $Q\subseteq V(G)$. 
 Let $P_a=\{i~|~x_1(i)=a\}$ for $a=0,1,\ldots, g_1-1$, be the vertices of $P$, 
 while $Q_b=\{i~|~x_2(i)=b\}$ for $b=0, 1, \ldots,g_2-1$, be the vertices of $Q$. We have that 
$\lfloor\frac{n}{g_1}\rfloor \leq |P_a|\leq  \lceil\frac{n}{g_1}\rceil$ and 
$\lfloor\frac{n}{g_2}\rfloor \leq |Q_b|\leq  \lceil\frac{n}{g_2}\rceil$, as $x_1$ and $x_2$ are balanced vectors.
For each $i=1,2,\ldots,n$ there exists exactly one $P_a\in P$ with $i\in P_a$ and exactly one $Q_b\in Q$ with $i\in Q_b$. 
For each such $i$, add an edge between vertices corresponding to $P_a$ and $Q_b$ and label it $i$. Hence $d_G(P_a)= |P_a|$ and 
$d_G(Q_b)= |Q_b|$. 
If any vertex $v$ of $G$ has $d_G(v) > h$ then we split it into $\lfloor \frac{d_G(v)}{h} \rfloor$ vertices
 of degree $h$ and, if necessary, one vertex of degree $d_G(v)- h\lfloor \frac{d_G(v)}{h} \rfloor $. Denote this resultant 
 bipartite multigraph by $H$ with maximum degree $\Delta (H)=h$.   
We know that a bipartite graph $H$ with maximum degree $h$ is the union of $h$ matching. Thus $E(H)$ is union of $h$ 
matchings $F_0$, $F_1$, \ldots, $F_{h-1}$. Now identify those points of $H$ which corresponds to the same point of $G$,
then $F_0$, $F_1$, \ldots, $F_{h-1}$ are mapped onto certain edge disjoint spanning  subgraphs 
$F^{\prime}_0$, $F^{\prime}_1$, \ldots, $F^{\prime}_{h-1}$ of $G$. 
These $h$ edge-disjoint spanning subgraphs $F^{\prime}_0$, $F^{\prime}_1$, \ldots, $F^{\prime}_{h-1}$ of $G$ 
form a partition of $E(G)=[1,n]$ which we use to build a balanced vector $y\in \mathbb Z_h^n$.
Each edge disjoint spanning subgraph corresponds to a symbol in $\mathbb Z_h$ and each edge corresponds to an index 
from $[1,n]$.
Suppose edge disjoint spanning subgraph $F^{\prime}_c$ corresponds to symbol
$c\in \mathbb Z_h$. For each edge $i$ in $F^{\prime}_c$, define $y(i)=c$.
 Since $F_i$ is a matching, there is atmost one $F_i$-edge incident with any of the $\lceil \frac{d_G(P_a)}{h}\rceil$ vertices 
 of $H$ corresponds to $P_a\in P$. Hence $$d_{F_i^{\prime}}(P_a)\leq \lceil \frac{d_G(P_a)}{h}\rceil.$$
On the other hand, there are $\lfloor \frac{d_G(P_a)}{h}\rfloor$ vertices of $H$ corresponds to $P_a$ which have degree $h$. 
There must be an $F_i$-edge starting from each of these, whence
$$d_{F_i^{\prime}}(P_a)\geq \lfloor \frac{d_G(P_a)}{h}\rfloor.$$ Thus  
we have  $ \lfloor\frac{n}{g_1h}\rfloor\leq d_{F_i^{\prime}}(P_a)\leq \lceil\frac{n}{g_1h}\rceil$ for $i=0,1,\ldots,h-1$. 
This means that there exist
$\lfloor\frac{n}{g_1h}\rfloor$ or $\lceil\frac{n}{g_1h}\rceil$ edges $i\in [1,n]$ such that
$x_1(i)=a$ and $y(i)=c$, or in other words, each pair of symbols 
$(a,c)\in \mathbb Z_{g_1}\times\mathbb Z_h$ between $x_1$ and $y$ appears either $ \lfloor\frac{n}{g_1h}\rfloor$ 
or $\lceil\frac{n}{g_1h}\rceil$ times.
So, $x_1$ and $y$ are pairwise balanced vectors. Similarly, we can show that $y$ and $x_2$ are pairwise balanced vectors. 
Next, we need to show that $y$ is balanced. 
This corresponds to each spanning subgraph $F^{\prime}_i$  
contains either $\lfloor\frac{n}{h}\rfloor$ or $\lceil\frac{n}{h}\rceil$ edges. In other words, this corresponds to 
each matching $F_i$ contains either $\lfloor\frac{n}{h}\rfloor$ or $\lceil\frac{n}{h}\rceil$ edges.   
Suppose  we have two matchings $F_0$ and $F_1$  that differ by size more than 1, say $F_0$ smaller and $F_1$ larger. 
Every component of the union of $F_0$ and $F_1$ could be an alternating even cycle or an alternating path. 
 Note that it must contain a path, otherwise their sizes are equal. We can find a path component in the union graph
that contains more edges from $F_1$ than $F_0$. Swap the $F_1$ edges with the $F_0$ edges in this path component. 
Then the resultant graph has $F_0$ 
increased in size by 1 edge, and $F_1$ decreased in size by 1 edge. Continue this process on 
$F_0$, $F_1$, \ldots, $F_{h-1}$  until the sizes are correct.    
\end{proof}

\noindent The following corollary is an easy consequence of Lemma \ref{lemma1}.

\begin{corollary}\label{cor1}
Let $ x \in \mathbb Z_{g}^{n}$ be  a balanced vector.
 Then for any positive integer $h$, there exists a balanced vector $y\in \mathbb Z_{h}^{n}$ 
such that $x$ and $y$ are pairwise balanced.
 \end{corollary}
\begin{proof}
This follows from Lemma \ref{lemma1}. Set $x_1=x$ and $x_2=x$. 
\end{proof}

\begin{lemma}\label{lemma2} Let $x_1\in \mathbb Z_{g_1}^n $ and $x_2\in \mathbb Z_{g_2}^n $ be two pairwise balanced vectors. 
Then for any $h$ such that $g_1g_2h\leq n$, there exists a balanced vector $y\in \mathbb Z_h^n $
such that $x_1$, $x_2$ and $y$ are 3-qualitatively independent and $x_1$ and $y$ are pairwise balanced and $x_2$ and $y$ are pairwise balanced.

\end{lemma}

\begin{proof}
Construct a bipartite multigraph $G$ corresponds to $x_1$ and $x_2$ as defined in the proof of Lemma \ref{lemma1}. We have that the vectors $x_1$ and $x_2$ are pairwise balanced, that is, for each pair $(a,b)\in\mathbb Z_{g_1}\times \mathbb Z_{g_2}$, the 
number of edges between $P_a$ and $Q_b$ is $\lfloor \frac{n}{g_1g_2}\rfloor$ or $\lceil \frac{n}{g_1g_2}\rceil$. The problem is to find a balanced vector 
$y\in \mathbb Z_h^n$, such that $x_1$, $x_2$ and $y$ are 3-qualitatively independent, $x_1$ and $y$ are pairwise balanced, and $x_2$ and $y$ are
pairwise balanced. Assume without loss of generality that $g_1\leq g_2$. We construct a bipartite multigraph $H$ from $G$ as follow:  
We split each point $P_a \in P$ in $G$  into $\lfloor \frac{d_G(P_a)}{h} \rfloor$ points of degree $h$ and, if necessary, one point of degree $d_G(P_a)- h\lfloor \frac{d_G(P_a)}{h} \rfloor $
 in $H$. Thus, using balancedness of $x_1$, we have that there are at least $g_2 $ copies of $P_a$ in $H$ from the split operation. 
Label them $P_{a0}$, $P_{a1}$, ..., $P_{a,g_2-1}$, 
$P_{a g_2} \ldots$ ($g_2$ onwards are extra). Similarly we split each point $Q_b\in Q$ into $\lfloor \frac{d_G(Q_b)}{h} \rfloor$ points of degree $h$ and, if necessary, 
one point of degree $d_G(Q_b)- h\lfloor \frac{d_G(Q_b)}{h} \rfloor $  in $H$. 
Thus, using balancedness of $x_2$, we have that there are at least $ g_1 $ copies of $Q_b$ in $H$ from the split operation. 
Label them $Q_{b0}$, $Q_{b1}$, ..., $Q_{b,g_1-1}$,
$Q_{bg_1} \ldots$ ($g_1$ onwards are extra). 
For each pair of vertices $P_a$ and $Q_b$, we have at least $h$ edges between $P_a$ and $Q_b$; consider only the first $h$ edges from $P_a$ to $Q_b$ (ignore the rest for now).
These $h$ edges between $P_a$ and $Q_b$ in $G$ become the $h$ edges between $P_{ab}$ and $Q_{ba}$ in $H$.
This results in a graph (possibly multigraph) where every vertex has maximum degree $h$.  We add remaining edges arbitrarily to $H$ amongst  the  remaining vertices 
 (including the extra vertices) in any way, provided we maintain $H$ as bipartite graph with maximum degree $h$ and every vertex
 $v$ of $G$ is split  into $\lfloor \frac{d_G(v)}{h} \rfloor$ points of degree $h$ and, if necessary, one point of degree $d_G(v)- h\lfloor \frac{d_G(v)}{h} \rfloor $.
We know that a bipartite graph with maximum degree $h$ is the union of $h$ matching. Thus $E(H)$ is union of $h$ 
matchings $F_0$, $F_1$, \ldots, $F_{h-1}$. Now identify those points of $H$ which corresponds to the same point of $G$, then $F_0$, $F_1$, \ldots, $F_{h-1}$
are mapped onto certain edge disjoint spanning  subgraphs $F^{\prime}_0$, $F^{\prime}_1$, \ldots, $F^{\prime}_{h-1}$ of $G$. 
We claim each of the spanning subgraphs $F^{\prime}_i$ is a complete bipartite multigraph. 
 For every $a\in \mathbb Z_{g_1}$, $b\in\mathbb Z_{g_2} $, there are $h$ edges from $P_{ab}$ to $Q_{ba}$ in $H$, 
and they will all appear in different matchings $F_0$, $F_1$, \ldots, $F_{h-1}$.
This ensures that the spanning subgraphs contain at least one $P_a-Q_b$ edge for every $a\in \mathbb Z_{g_1}$, $b\in\mathbb Z_{g_2} $.
This proves that each of the spanning subgraphs $F^{\prime}_i$ is a complete bipartite multigraph.  
 These $h$ edge-disjoint spanning subgraphs $F^{\prime}_0$, $F^{\prime}_1$, \ldots, $F^{\prime}_{h-1}$ of $G$ form a partition of $E(G)=[1,n]$ which we use to build a balanced vector $y\in \mathbb Z_h^n$.
Each edge disjoint spanning subgraph corresponds to a symbol in $\mathbb Z_h$ and each edge corresponds to an index from $[1,n]$.
Suppose edge disjoint spanning subgraph $F^{\prime}_c$ corresponds to symbol
$c\in \mathbb Z_h$. For each edge $i$ in $F^{\prime}_c$, define $y(i)=c$. We need to show that $x_1$, $x_2$, $y$ are 3-qualitatively independent. 
For any  $a\in \mathbb Z_{g_1}$, $b\in \mathbb Z_{g_2}$, $c\in \mathbb Z_h$, in the spanning subgraph $F^{\prime}_c$ there is an edge $i$ incident to
$P_a\in P$ and $Q_b\in Q$ as $F^{\prime}_c$ is a complete bipartite multigraph. This means that for any $a\in \mathbb Z_{g_1}$,
$b\in \mathbb Z_{g_2}$, $c\in \mathbb Z_h$, there exists an $i\in [1,n]$ such that $x_1(i)=a$, $x_2(i)=b$, and $y(i)=c$. So, $x_1$, $x_2$ and $y$ are 3-qualitatively independent.
Next, we prove that $x_1$ and $y$ are pairwise balanced, and $x_2$ and $y$ are
pairwise balanced. Since $F_c$ is a matching, there is atmost one $F_c$-edge incident with any of the $\lceil \frac{d_G(P_a)}{h}\rceil$ vertices of $H$ corresponds to $P_a\in P$. Hence
$$d_{F_c^{\prime}}(P_a)\leq \lceil \frac{d_G(P_a)}{h}\rceil.$$
On the other hand, there are $\lfloor \frac{d_G(P_a)}{h}\rfloor$ points of $H$ corresponds to $P_a$ which have degree $h$. There must be an $F_c$-edge starting from each of these, whence
$$d_{F_c^{\prime}}(P_a)\geq \lfloor \frac{d_G(P_a)}{h}\rfloor.$$ Thus  
we have  $ \lfloor\frac{n}{g_1h}\rfloor\leq d_{{F_c}^{\prime}}(P_a)\leq \lceil\frac{n}{g_1h}\rceil$ for $c=0,1,\ldots,h-1$. This means that there exist
$\lfloor\frac{n}{g_1h}\rfloor$ or $\lceil\frac{n}{g_1h}\rceil$ edges $i\in [1,n]$ such that $x_1(i)=a$ and $y(i)=c$, or in other words, each pair of symbols $(a,c)\in
\mathbb Z_{g_1}\times\mathbb Z_h$ between $x_1$ and $y$ appears either $ \lfloor\frac{n}{g_1h}\rfloor$ or $\lceil\frac{n}{g_1h}\rceil$ times.
So, $x_1$ and $y$ are pairwise balanced vectors. Similarly, we can show that $y$ and $x_2$ are pairwise balanced vectors. Next, we need to show that $y$ is balanced. 
This corresponds to each spanning subgraph $F_c^{\prime}$  contains either $\lfloor\frac{n}{h}\rfloor$ or $\lceil\frac{n}{h}\rceil$ edges. In other words, this corresponds to 
each matching $F_c$ contains either $\lfloor\frac{n}{h}\rfloor$ or $\lceil\frac{n}{h}\rceil$ edges.   The proof of balancedness is  the same as that of Lemma \ref{lemma1}.

\end{proof}

\section{Optimal Mixed Covering Array on 3-Uniform Hypergraph}
\noindent 
Let $ H $ be a vertex-weighted 3-uniform hypergraph with $ k $ vertices. Label the vertices $ v_{1}, v_{2},..., v_{k}$ and for
each vertex $ v_{i} $ denote its associated weight by $ w_{H}(v_{i}) $. Let the $ \textit{product weight} $ of $ H $, 
denoted $ PW(H) $, be
$$PW(H) = \mbox{max}  \lbrace w_H(u)w_H(v)w_{H}(w) : \{u, v, w\} \in E(H) \rbrace.$$
 Note that 3-$CAN( H, \prod_{i=1}^{k} w_{H}(v_{i}) ) \geq  PW(H)$. A balanced covering array on  $H$ is 
 a covering array on $H$ in which every row is balanced and the rows correspond to vertices in a hyperedge are pairwise balanced. 
 
\subsection{Basic Hypergraph Operations}
 We now introduce four hypergraph operations:
\begin{enumerate}
\item \textit{Single-vertex edge hooking I}
\item \textit{Single-vertex edge hooking II}
\item \textit{Two-vertex hyperedge hooking}
\item \textit{Single-vertex hyperedge hooking I}
\end{enumerate}

 \begin{figure}[h]
 \centering
 \includegraphics[width=12cm]{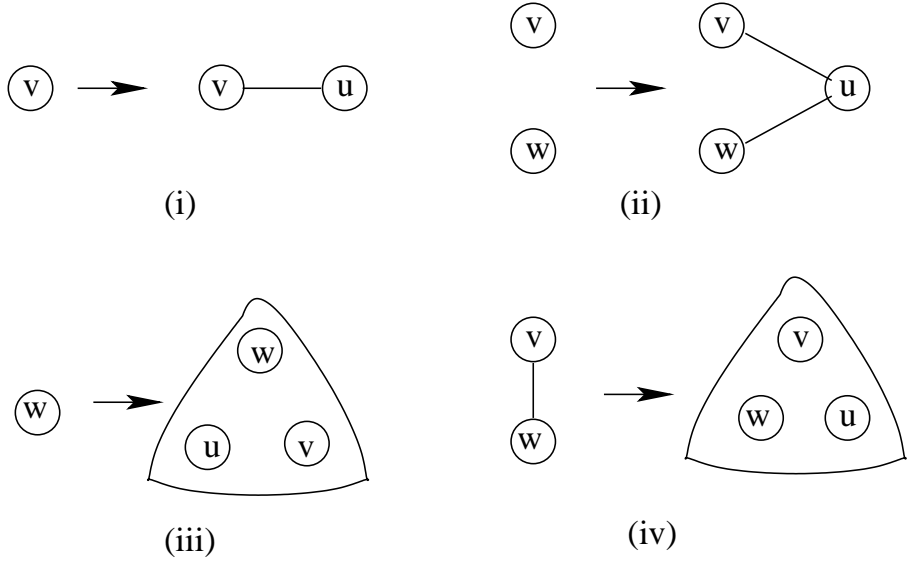}
 \caption{(i) Single-vertex edge hooking I (ii) Single-vertex edge hooking II  (iii) Two-vertex hyperedge hooking (iv) Single-vertex hyperedge hooking I}
 \label{fig:3}
\end{figure}

 A {\it single-vertex edge hooking I} in hypergraph $H$ is the operation that  inserts a new edge $\{u,v\}$ in which $ u$ is a new vertex and $v$ is in $V(H)$.  
A {\it single-vertex edge hooking II} in hypergraph $H$ is the operation that  inserts two new edges $\{u,v\}$  and $\{u,w\}$ in which $ u$ is a new vertex and 
$v$ and $w$ are in $V(H)$.  
A {\it two-vertex hyperedge hooking} in a hypergraph $H$ is the operation that insert a new hyperedge $\{u,v,w\}$ in which $u$ and $v$ are new vertices and $w$ is 
in $V(H)$. 
  A  {\it single vertex hyperedge hooking I} in a  hypergraph $H$ 
 is the operation that  replaces an edge $\{v,w\}$  by a  hyperedge $\{u,v,w\}$ where $u$ is a new vertex.

\begin{proposition} \label{prop1} Let $H$ be a weighted hypergraph with $k$ vertices and $H'$ be the weighted 
hypergraph obtained from $H$ by single-vertex edge hooking I, single-vertex edge hooking II or single vertex hyperedge hooking I operation  with $u$ as a new vertex with $w(u) $ such that $PW(H)=PW(H^{\prime})$. 
Then, there exists a balanced $CA( n, H, \prod_{i=1}^{k} g_{i})$ 
 if and only if there exists a balanced $CA(n, H^{'}, w(u)\prod_{i=1}^{k} g_{i})$.
\end{proposition}
\begin{proof}
  If there exists a balanced CA($ n, H', w(u)\prod_{i=1}^{k} g_{i} $) then by deleting the row 
 corresponding to the new vertex $u$ we can obtain a $CA(n, H, \prod_{i=1}^{k} g_{i})$. 
 Conversely, let $ C^{H} $ be a balanced CA($ n, H, \prod_{i=1}^{k} g_{i} $). The balanced covering array  $ C^{H} $ can be used to construct  
 $ C^{H^{\prime}}$, a balanced CA($ n, H^{'}, w(u)\prod_{i=1}^{k} g_{i} $). We consider the following cases:\\
 \noindent Case 1:  Let $H^{\prime}$ be obtained from $H$ by a single vertex edge hooking I 
 of a new vertex $u$ with a new edge $\{u,v\}$, and $w(u)$ such that $w(u)w(v)\leq n$.  Using Corollary \ref{cor1}, we can build a balanced length-$n$ vector $y$ corresponds to vertex $u$ such that  $y$
 is pairwise balanced with the length-$n$ vector $x$ corresponds to vertex $v$.  
 The array $C^{H^{\prime}}$ is built by appending row $y$ to $C^H$. \\
 Case 2: Let $H^{\prime}$ be obtained from $H$ by a single vertex edge hooking II 
 of a new vertex $u$ with two new edges $\{u,v\}$ and $\{u,w\}$, and $w(u)$ such that $w(u)w(v)\leq n$ and $w(u)w(w)\leq n$.  Using Lemma \ref{lemma1}, we can build a balanced length-$n$ vector $y$ corresponds to vertex $u$ such that  $y$
 is pairwise balanced with the length-$n$ vectors $x_1$ and $x_2$ correspond to vertices $u$ and $v$ respectively.  
 The array $C^{H^{\prime}}$ is built by appending row $y$ to $C^H$. \\
 Case 3: If $H^{\prime}$ is obtained from $H$ by  replacing an edge $\{v,w\}\in E(H)$ by a new hyperedge $\{u,v,w\}$ in which $u$ is a new vertex, 
 and $w(u)$ such that $w(u)w(v)w(w)\leq n$. 
 Using Lemma \ref{lemma2}, we can build a balanced length $n$ vector $y$ corresponds to vertex $u$ such that $y$ is 3-qualitatively independent with two length-$n$
 pairwise balanced vectors $x_1$ and $x_2$ correspond to vertices $v$ and $w$ in $H$. The array $C^{H^{\prime}}$ is built by appending row $y$ to $C^H$.
  \end{proof}
  
  \begin{proposition} \label{new} Let $H$ be a weighted hypergraph with $k$ vertices and $H'$ be the weighted 
hypergraph obtained from $H$ by two-vertex hyperedge hooking operation  with $u$ and $v$ as  new vertices  with $w(u) $ and $w(v)$  such that $PW(H)=PW(H^{\prime})$. 
Then, there exists a balanced $CA( n, H, \prod_{i=1}^{k} g_{i})$ 
 if and only if there exists a balanced $CA(n, H^{'}, w(u)w(v)\prod_{i=1}^{k} g_{i})$.
\end{proposition}

\begin{proof}
If there exists a balanced CA($ n, H', w(u)\prod_{i=1}^{k} g_{i} $) then by deleting the rows 
 corresponding to the new vertices $u$ and $v$  we can obtain a $CA(n, H, \prod_{i=1}^{k} g_{i})$. 
 Conversely, let $ C^{H} $ be a balanced CA($ n, H, \prod_{i=1}^{k} g_{i} $). 
  Hypergraph $H^{\prime}$ is obtained from $H$ by a two-vertex hyperedge hooking  
 of two new vertices $u$ and $v$  with  a new hyperedge  $\{u,v,w\}$, and $w(u)$, $w(v)$ such that $w(u)w(v)w(w)\leq n$. Using Corollary \ref{cor1}, we can build a balanced length-$n$ vector $y_1$ corresponds to vertex $u$ such that  $y_1$
 is pairwise balanced with the length-$n$ vector $x$ corresponds to vertex $w$.  Then using Lemma \ref{lemma2}, we can build a balanced length $n$ vector $y_2$ corresponds to vertex $v$ such that $y_2$ is 3-qualitatively independent with two length-$n$
 pairwise balanced vectors $x$ and $y_1$ correspond to vertices $w$ and $u$ respectively in $H$. The array $C^{H^{\prime}}$ is built by appending rows 
 $y_1$ and $y_2$ to $C^H$.\\ 
 \end{proof}
 
 \begin{theorem}\label{thm1} Let $ H $ be a weighted  hypergraph and $ H^{'}$ be a weighted 3-uniform hypergraph obtained 
from $ H $ via a sequence of single-vertex edge hooking I, single-vertex edge hooking II,  two-vertex hyperedge hooking, single-vertex hyperedge hooking I operations.
 Let $ v_{k+1}, v_{k+2},..., v_{l} $ be the vertices in
$ V(H^{'})\setminus V(H)$ with weights $g_{k+1}, g_{k+2},..., g_{l}$ respectively so that $PW(H)=PW(H^{\prime})$. If there exists a balanced covering array 
 $CA(n,H,\prod_{i=1}^{k}g_{i})$, then there exists a balanced $CA( n, H^{'}, \prod_{i=1}^{l} g_{i} )$.
\end{theorem}
\begin{proof} The result is derived by iterating the different cases of Proposition \ref{prop1} and Proposition \ref{new}. \end{proof}

\subsection{$\alpha$-acyclic 3-uniform hypergraphs}
The notion of hypergraph acyclicity plays crucial role in numerous fields of application of hypergraph theory specially in 
 relational database theory and constraint programming. There are many generalizations of the notion  of graph  acyclicity in 
 hypergraphs.  Graham \cite{G},
  and independently, Yu and Ozsoyoglu \cite{Y}, defined $\alpha$-acyclic  property for hypergraphs via a transformation now known as 
  the {\it GYO reduction}.  Given a hypergraph $H=(V,E)$, the GYO reduction applies the following operations repeatedly to $H$ until none can be applied: 
  \begin{enumerate}
  \item If a vertex $v\in V$ has degree one, then delete $v$ from the edge containing it.
  \item If $A,B\in E(H)$ are distinct hyperedges such that $A\subseteq B$, then delete $A$ from $E(H)$.
  \item If $A\in E$ is empty, that is $A=\phi$, then delete $A$ from $E$. 
  
  \end{enumerate}
  
\begin{definition}
A hypergraph $H$ is $\alpha$-acyclic if GYO reduction on $H$ results in an empty hypergraph. 

\end{definition}

 \begin{example}
 Hypergraph $H_1= (V,E)$ with $V=\{1,2,3,4,5,6\}$ and $E=\{\{1,2,3\}, \{1,3,4\},\{1,2,6\},\{2,3,5\}\}$ is $\alpha$-acyclic.  
 \end{example}
 
 \begin{example}
 Hypergraph $H_2= (V,E)$ with $V=\{1,2,3,4,5,6\}$ and $E=\{\{1,2,3\}, \{1,3,4\}, \{2,4,5\}, \{4,5,6\}\}$ is not $\alpha$-acyclic.  
 \end{example}

\begin{theorem}\label{thm3}
 Let $ H $ be a weighted $\alpha$-acyclic 3-uniform hypergraph with $ l $ vertices. Then there exists a balanced mixed
 3-$CA( n, H, \prod_{i=1}^{l} g_{i} )$ with $n = PW(H)$. \end{theorem}

\begin{proof} Apply the GYO reduction on $H$ to record the order in which the hyperedges are deleted. Let $e_1,e_2,\ldots, e_m$ be a  deletion order for  the $m$ hyperedges of $H$.
 While constructing covering array on 
$H$, consider the hyperedges in  reverse order of their deletions.  
Let $H_1$ be the hypergraph with the single  hyperedge $e_m=\{v_1,v_2,v_3\}$. 
If $g_1g_2g_3=n$, there exists a trivial balanced covering array $CA(n, H_1, \prod_{i=1}^3{g_i})$.   Otherwise, 
if $g_1g_2g_3\leq n$,   we construct a balanced covering array of size $n$ on $H_1$ as follows: begin with a balanced vector $x_1\in \mathbb Z_{g_1}^n$ corresponds to vertex $v_1$.  From Proposition \ref{new} (using two-vertex hyperedge hooking operation), we get a balanced covering array $CA(n, H_1, \prod_{i=1}^3{g_i})$.   Let $H_2$ be the hypergraph obtained from $H_1$ by adding hyperedge $e_{m-1}$. 
Using single-vertex hyperedge hooking I or two-vertex hyperedge hooking operation, there exists a covering array of size $n$ on $H_2$. For $i=2,3,\ldots,m$, let $H_i=H_{i-1}\cup e_{m+1-i}$. Note that $H_m=H$.  As  $PW(H_i)\leq PW(H)$ for all $i=2,3,\ldots,m$, using single-vertex hyperedge hooking I or two-vertex hyperedge hooking operation, there 
exists a balanced covering array on $H_i$ of size $n$. In particular, there exists a balanced  3-$CA( n, H, \prod_{i=1}^{l} g_{i} )$.
\end{proof}

\begin{definition} \cite{vitaly}  A hypergraph $ H = (V, E) $ is called an  interval hypergraph if there exists a linear 
ordering of the vertices $ v_{1}, v_{2},..., v_{n} $ such that every hyperedge of $ H $ induces an interval in this ordering. 
In other words, the vertices in $ V $ can be placed on the real line such that every hyperedge is an interval.
\end{definition}

\begin{corollary} Let $ H $ be a weighted 3-uniform interval hypergraph with $l$ vertices.
 Then there exists a balanced mixed
 3-$CA( n, H, \prod_{i=1}^{l} g_{i} )$ where $n=PW(H)$.  
 \end{corollary}
\begin{proof} This corollary follows immediately from the proof of Theorem \ref{thm3} since every interval hypergraph is $\alpha$-acyclic. 
\end{proof}

\subsubsection{3-Uniform Hypertrees}

In this subsection, we give a construction for optimal mixed covering arrays on some specific conformal 3-uniform hypertrees. 
A $ \textit{host graph} $ for a hypergraph is a connected graph on the same vertex set, such that every hyperedge induces 
a connected subgraph of the host graph \cite{vitaly}.

\begin{definition} \rm(Voloshin \cite{vitaly}). \it A hypergraph $H = (V, E)$ is called a $ hypertree $ if there exists 
a host tree $ T = (V, E^{'}) $ such that each hyperedge $ E_{i} \in \textit{E} $ induces a subtree of $ T $.
\end{definition}
\noindent In other words, any hypertree is isomorphic to some family of subtrees of a tree. 
A 3-uniform hypertree is a hypertree such that each hyperedge in it contains exactly three vertices.

\begin{theorem}\label{thm2}
Let $ H $ be a weighted conformal 3-uniform hypertree with $ l $ vertices, having a binary tree as a host
tree. Then there exists a balanced mixed 3-$CA( n, H, \prod_{i=1}^{l} g_{i} )$ with $ n = PW(H) $. 
\end{theorem}
\begin{figure}[h]
 \centering
 \includegraphics[height=5cm]{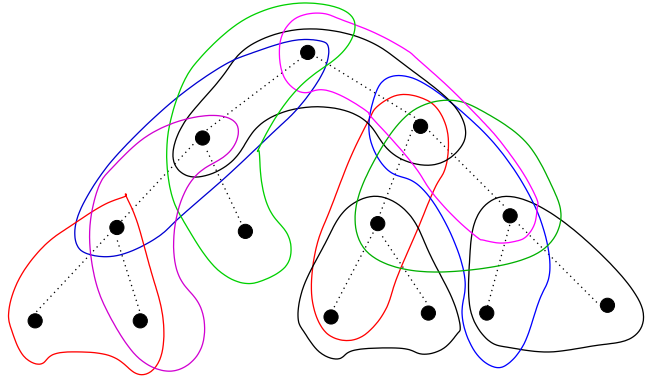}
 \caption{A conformal 3-uniform hypertree with a binary host tree}
 \label{fig:1}
\end{figure}

\begin{proof} We claim that $H$ is an $\alpha$-acyclic hypergraph. The reason is this. We do not have three hyperedges in $H$ with a common 
vertex and other 3 vertices pair-wise adjacent  as  conformality implies a hyperedge of size 4. Thus, $H$ has at least one vertex of degree 1. 
Apply the GYO reduction on $H$. At each iteration of the GYO reduction, it produces  a partial hypertree which is again a conformal 3-uniform hypertree having a binary tree as host 
tree. The GYO reduction on
$H$ results in an empty hypertree. Therefore, $H$ is an $\alpha$-acyclic hypergraph. Now the proof follows directly from the proof of Theorem \ref{thm3}.
\end{proof}

\subsection{3-uniform Cycles}
The cyclic structure is very rich in hypergraphs as compare to that in graphs \cite{bee}. It seems difficult  to 
construct optimal size mixed covering arrays on  cycle hypergraphs.  There are few special types of 3-uniform cycles for which we  construct optimal size mixed 
covering arrays.

\begin{theorem}\label{thm4}
 Let $H$ be a weighted 3-uniform cycle $(v_1,E_1,v_2,E_2,...,v_k,E_k,v_1)$ of length $k\geq 3$ on $2k$ vertices satisfying 
 the following conditions.
 \begin{enumerate}
  \item $E_i \cap E_{i+1}= \{v_{i+1}\}$ for $i=1,..,k-1$ and $E_k \cap E_1= \{v_1\}$
  \item $d(u_i)=1$ for $u_i\in E_i\smallsetminus \{v_i ,v_{i+1}\}$ where $i= 1,...,k-1$ and $d(u_k)=1$ for 
  $u_k\in E_k\smallsetminus \{v_k ,v_1\}$
 \end{enumerate}
 
 \begin{figure}[h]
 \centering
 \includegraphics[width=5.5cm]{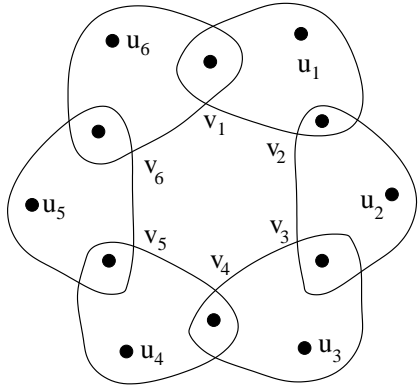}
 \caption{3-uniform cycle of length-6}
 \label{fig:2}
\end{figure}
\noindent Let $g_i$  and $\omega_i$ denote the weights of vertices $v_i$ and  $u_i$ respectively.  Then there exists a balanced 
 3-$CA( n, H, \prod_{j=1}^{k} g_{j}\omega_j )$ with $n = PW(H)$. 
\end{theorem}

\begin{proof} Let $\{v_1,u_1,v_2\}$ be a hyperedge in $H$ with  $g_1\omega_1g_2=PW(H)$. Let $H_1$ be the hypergraph with the single hyperedge $\{v_1,u_1,v_2\}$. 
There exists a balanced covering array 3-$CA(n,H_1, \omega_1\prod_{j=1}^2 g_j)$. 
For $i=2,3,\ldots, k-1$, let $H_i$ be the  hypergraph obtained from $H_{i-1}$ after inserting a  new edge $\{v_i,v_{i+1}\}$ in which $v_{i+1}$ is a new vertex, that is, $H_i=H_{i-1}\cup\{v_i,v_{i+1}\}$.  Using Proposition \ref{prop1} (single-vertex edge hooking I operation), for all $i=2,3,\ldots, k-1$, as $g_ig_{i+1} \leq n$, there exists a balanced 
$CA(n, H_i, \omega_1\prod_{j=1}^{i+1} g_j)$. Let $H_k=H_{k-1}\cup \{\{v_{k-1},v_k\}, \{v_k,v_1\}\}$.  Using single vertex edge hooking II operation, as $g_{k-1}g_k\leq n$ and 
$g_1g_k\leq n$, we get a balanced covering array $CA(n, H_k, \omega_1\prod_{j=1}^{k} g_j)$.  Finally, using sequence of single-vertex hyperedge hooking I operations on $H_k$,
 replace edge $\{v_i,v_{i+1}\}$
by hyperedge $\{v_i,u_i,v_{i+1}\}$ for $i=2,3,\ldots, k-1$; also replace edge $\{v_k,v_1\}$ by hyperedge $\{v_k,u_k,v_1\}$. As $g_i\omega_ig_{i+1}\leq n$ for all $i=2,3,\ldots,k-2$ and 
$g_k\omega_kg_1\leq n$, from Proposition \ref{prop1} (using single-vertex hyperedge hooking I), there exists a balanced 3-$CA( n, H, \prod_{j=1}^{k} g_{j}\omega_j )$.  
\end{proof}

\noindent The length-$k$ 3-uniform cycle considered in Theorem \ref{thm4} contains $k$ vertices of degree 1. As every hyperedge has 
one vertex of degree 1, such hypergraph satisfies $|E(H)|= |V(H)|/2$.

\section{Further Cycle Hypergraphs} In this section,  we consider 3-uniform cycles of length $k$ with exactly one vertex of degree  1. This type of 3-uniform 
hypergraphs have $|E(H)|= |V(H)|-2$.  Construction of  
optimal size mixed covering arrays on such cycle hypergraphs seems to be difficult. 

Let $H$ be a weighted 3-uniform cycle $(v_0,E_1,v_2,E_2,v_3,E_3,v_0)$ of length-3 on five vertices 
 with $E_1=\{v_0,v_1,v_2\}$, $E_2=\{v_{1},v_2,v_{3}\}$ and $E_3=\{v_3,v_4,v_0\}$ as shown in Figure \ref{diffcycle}. 
  Let $E_1$ be a hyperedge in $H$ with $g_0g_1g_2=PW(H)$ where $g_i$ denotes the weight of vertex $v_i$. Let $H_1$ be the hypergraph with the single hyperedge $E_1$. There exists a balanced covering array $CA(n,H_1, \prod_{i=0}^2 g_i)$ where $n=PW(H)$.
Let $H_2=H_1\cup\{E_2\}$.  Using Proposition \ref{prop1} (single-vertex hyperedge hooking I), there exists a balanced covering array   $CA(n,H_2, \prod_{i=0}^3 g_i)$. Let
$H_3=H_2 \cup \{E_3\}$. Note that $H_3=H$. We cannot use any of the known hypergraph operations to construct a balanced covering array of size $PW(H)$ on $H_3$ as the 
rows correspond to $v_0$ and $v_3$ are not pairwise balanced.   Thus we define a new hypergraph operation called {\it single vertex hyperedge hooking} II operation.  A  single vertex hyperedge hooking II in a  hypergraph $H$ 
 is the operation that inserts a new hyperedge $\{u,v,w\}$  and a new edge $\{u,z\}$ where $\{v,w,z\}$ is an existing hyperedge in $H$ and $u$ is a new vertex.   
 
\begin{figure}[h]
\label{diffcycle}
 \centering
 \includegraphics[height=4cm]{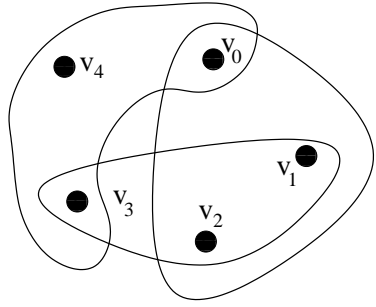}
 \caption{cycle of lenghth 3 with $g_0=10, g_1=8, g_2=5, g_3= 2, g_4=18$}
 \label{fig: cycle of length3}
\end{figure}

\subsection{Balanced Partitioning}
Let $ g_{1}, g_{2}, g_{3}$ and $n\geq g_1g_2g_3$ be positive integers and $ x_{1} \in \mathbb Z_{g_{1}}^{n}$, $ x_{2} \in \mathbb Z_{g_{2}}^{n} $ 
and $ x_{3} \in \mathbb Z_{g_{3}}^{n} $ be mutually pairwise balanced  and 3-qualitatively independent vectors. Then, we prove in this section, there
exists  a balanced vector $y\in \mathbb Z_{h}^{n}$, where $h$ satisfies certain conditions, such that $\lbrace x_{1} , x_{2} , y \rbrace$ are 3-qualitatively 
and $y$ is pairwise balanced with each $x_i$ for $i=1,2,3$. 

 We construct a tripartite 3-uniform multi-hypergraph $G$ corresponds to $x_1$,  $x_2$ and $x_3$ as follows: $G$ has $g_1$ vertices 
in the first part $P\subseteq V(G)$, $g_2$ vertices in the second part $Q\subseteq V(G)$ and 
$g_3$ vertices in the third part $R\subseteq V(G)$. 
 Let $P_a=\{i~|~x_1(i)=a\}$ for $a=0,1,\ldots, g_1-1$, be the vertices of $P$, 
 $Q_b=\{i~|~x_2(i)=b\}$ for $b=0, 1, \ldots,g_2-1$, be the vertices of $Q$, and 
 $R_c=\{i~|~x_3(i)=c\}$ for $c=0, 1, \ldots,g_3-1$, be the vertices of $R$. For each $i=1,2,\ldots,n$ there exists exactly one $P_a\in P$ with $i\in P_a$, exactly one $Q_b\in Q$ with $i\in Q_b$ and 
exactly one $R_c\in R$ with $i\in R_c$. For each such $i$, add a hyperedge $\{P_a,Q_b,R_c\}$ and label it $i$.  Clearly, $d_G(P_a)=|P_a|$,  $d_G(Q_b)=|Q_b|$ and $d_G(R_c)=|R_c|$.  Let $h$ be a
positive integer so that $h\leq \min\{\lfloor\frac{n}{g_1g_2}\rfloor, \lfloor\frac{n}{g_1g_3}\rfloor\}$ and $$\lfloor\frac{n}{g_1g_2}\rfloor \equiv 0  \mbox{~~~mod~~} h~~~\mbox{ for} ~h\geq 3.$$ That is, for each pair $(a,b)\in \mathbb Z_{g_1} \times \mathbb Z_{g_2}$, the number $d_G(P_aQ_b)$ of hyperedges containing $P_a$ and $Q_b$ is either $0$ or $1$ mod $h$. Clearly,
$d_G(P_aQ_b)=|P_a\cap Q_b|$.
We construct a tripartite 3-uniform hypergraph $H$ with maximum degree $h$ from $G$ as follows: We split each vertex $v\in V(G)$ in $G$ into 
$\lfloor\frac{d_G(v)}{h} \rfloor$ vertices of degree $h$ and, if necessary, one vertex of degree less than $h$.  Using balancedness of $x_1$, we have that
there are at least $g_2$ copies of $P_a$ in $H$ from the split operation. Label them $P^{l}_{a0}, \dots P^{l}_{a,g_2-1}, \mathcal{E}_a^1, \mathcal{E}_a^2,\ldots$ for $l=1,2,\ldots,
\lfloor\frac{d_G(P_aQ_b)}{h}\rfloor$.  Similarly, there are at least
$g_1$ copies of $Q_b$, label them  $Q^{l}_{b0}, \dots Q^{l}_{b,g_1-1}, \mathcal{F}_b^1, \mathcal{F}_b^2,\ldots$  for $l=1,2 \ldots, \lfloor\frac{d_G(P_aQ_b)}{h}\rfloor$ and at least $g_1$ copies of $R_c$; label them 
$R^l_{c0},\ldots,R^l_{c,g_1-1}$, $\mathcal{G}_c^1,\mathcal{G}_c^2,\ldots$ for $l=1,2 \ldots, \lfloor\frac{d_G(P_aR_c)}{h}\rfloor$. \\

Each $P_a$ is split as follows: We have either $sh$ or $sh+1$ hyperedges containing $P_a$ and $Q_b$ for $b=0,1,\ldots,g_2-1$ where $s=\lfloor \frac{d_G(P_aQ_b)}{h}\rfloor$.  
Choose a $c\in Z_{g_3}$ (not necessarily distinct for different $a$). 
If the number of hyperedges containing $P_a$ and $Q_b$ is $sh+1$, we pick one hyperedge  $i\in P_a\cap Q_b$ so that $x_3(i)=c$. This is possible as $x_1,x_2,x_3$ are
3-qualitatively independent.  Let $E_{a}$ be the collection of all those hyperedges for 
$b=0,1,\ldots,g_2-1$; clearly $|E_a|\leq g_2$. Split 
$E_{a}$ into $\lfloor \frac{|E_{a}|}{h}\rfloor$ vertices of degree $h$ and, if necessary, one vertex of degree less than $h$. Denote these vertices as $\mathcal{E}_{a}^l$ for $l=1,2,\ldots, \lfloor \frac{|E_a|}{h}\rfloor+1$. Beside the hyperedges in $E_{a}$, we have exactly $sh$ hyperedges containing  $P_a$ and $Q_b$ . These $sh$
hyperedges  are partitioned into $s$ equal parts.  The $h$ hyperedges in one part become $h$ hyperedges containing $P^{l}_{ab}$ and $Q^{l}_{ba}$, $l=1,2,\ldots,s$, in $H$.  \\
Each $Q_b$ is split as follows: For $a=0,1,\ldots,g_1-1$, set $Q^{l}_{ba}=P^{l}_{ab}$. Distribute the remaining elements of $Q_b$ into vertices of degree $h$ and, if necessary, one vertex of degree less than $h$.  Denote these vertices as $\mathcal{F}_b^l$.  Each $R_{c}$ is split as follows:   $R_{c}$ is split so that $ \mathcal{E}_a^l  \subseteq R^{l}_{ca}$.   Distribute the remaining elements of $R_c$ into vertices of degree $h$ and, if necessary, one vertex of degree less than $h$. It is easy to observe that this partitioning of $P_a$, $Q_b$ and $R_c$ is not uniquely determined.

\begin{lemma}\label{BP}
$H$  is  balanced hypergraph with maximum degree $\Delta(H)= h$.
\end{lemma}

\begin{proof}
Hypergraph $H$ has $V(H)=X_1\cup X_2 \cup X_3$ as vertex set where 
$X_1=\{P_{ab}^l, \mathcal{E}_a^l \mid a\in \mathbb Z_{g_{1}}, b\in \mathbb Z_{g_{2}}, l\in \mathbb{N}\}$, 
$X_2=\{Q_{ba}^l, \mathcal{F}_b^l \mid a\in \mathbb Z_{g_{1}}, b\in \mathbb Z_{g_{2}}, l\in \mathbb{N}\}$ 
and $X_3=\{R_{ca}^l, \mathcal{G}_c^l \mid a\in \mathbb Z_{g_{1}}, c\in \mathbb Z_{g_{3}}, l\in \mathbb{N}\}$. 
 Let $A\subset V(H)$ and $H_A$ be the subhypergraph induced by $A$.  From Theorem \ref{2color}, it suffices to show that $H_A$ is 2-colourable. Later part of proof 
deals with 2-colouring of $H_A$ which is based on the following cases.\\
\noindent\textit{Case 1:}  $A\cap X_i =\emptyset$ for two choices of $i\in \{1,2,3\}$. Without loss of generality we assume 
$A\cap X_1 =\emptyset$ and $A\cap X_2 =\emptyset$, that is, $A$ intersects only with $X_3$. Being $H$ a tripartite hypergraph,
$A$ is an independent set in this case and $H_A$ has no hyperedges. Hence it is 2-colourable.
\\\textit{Case 2:}  $A\cap X_i =\emptyset$ for exactly one $i$. Without loss of generality we assume 
$A\cap X_1 =\emptyset$. As $A$ intersects with  $X_2$ and $X_3$, the induced sub-hypergraph $H_A$ is a bipartite graph between $X_2$ and $X_3$. 
Hence $H_A$ is 2-colourable.
\\\textit{Case 3:}  $A\cap X_i \neq \emptyset$ for all $i$. We claim that $H_A$ is union of a 3-uniform partial hypergraph of $H$ 
and a bipartite graph on $A$. Every partial hypergraph of $H$ is 2-colourable as $H$ is 2-colourable. Consider a 2-colouring of 
bipartite graph induced by subhypergraph and extend this to 2-colouring of 3-uniform partial hypergraph to produce a 2-colouring of $H_A$. 
To show that subgraph induced by $A$ is a bipartite graph consider a 2-uniform cycle $C$ in $H_A$.
If $C$ does not intersect some $X_i$ then it alternates between vertices of only two 
partite sets and 
turns out as a bipartite graph.  Now we assume $C$ intersects each partite set
$X_1, X_2$  and $X_3$.  Consider a vertex $v \in C\cap X_1$. We denote by $N_H(v)$  the set of neighbours of $v$ in $H$. 
There are two types of vertices in $X_1$ either of the form $P^{l}_{ab}$ or of the form $E_a^l$. If $v$ is 
$P^l_{ab}$ then $N_H(v)\cap X_2$ has only one vertex which is  $Q^{l}_{ba}$. 
Hence the edge $P^{l}_{ab}Q^{l}_{ba}$ cannot 
be  part of any cycle in $H_A$. Consequently both neighbours  of $P^{l}_{ab}$ in $C$ are from $X_3$ and corresponding incident 
edges in $C$ are induced only if $Q^l_{ba} \notin A$. If $v$ is $\mathcal{E}_a^l$ then $N_H(v)\cap X_3$ has only one vertex 
which is some $R^l_{ca}$.  Hence the edge $\mathcal{E}_a^lR^l_{ca}$ cannot be  part of  cycle  $C$. Consequently both 
neighbours of $\mathcal{E}_a^l$  in $C$ are from $X_2$ and corresponding incident edges in $C$ are induced only if 
$R^l_{ca} \notin A$. Thus either $N_C(v)\subset X_2$ or  $N_C(v)\subset X_3$. 
We identify the neighbours  $N_C(v)\in X_i$ as a  single vertex $N(v)$ from $X_i$. This identification operation reduces the length
of $C$ by  two and creates a smaller cycle with $v$ hanging out side of this new cycle by an edge incident at 
$N(v)$ with multiplicity 2. After performing identification for each $v\in C\cap X_1$, we left with a cycle $C'$ that 
alternates between vertices in $X_2$ and $X_3$. Consequently $C'$ has to be of  even length. 
Each identification operation reduces the length of $C$ by 2 whence total 
reduction in length is even. The length of $C$ is equal to sum of length of $C'$ and the total reduction and hence it is an 
even integer. This shows that $H_A$ does not contain any odd length 2-uniform cycle.   
\end{proof}

\begin{definition}\cite{berge}
 A matching in a hypergraph $H$ is a family of pairwise disjoint hyperedges. In other words matching is a partial hypergraph 
 $H_0$ with maximum degree $\Delta (H_0)= 1$.
\end{definition}

\begin{theorem}\label{balanced}\cite{kapoor}
 The hyperedges of a balanced hypergraph $ H $ with maximum degree $ \Delta$, can be partitioned into $ \Delta $ matchings.
\end{theorem} 

\begin{lemma}\label{lemma3}
Let $ x_{1} \in \mathbb Z_{g_{1}}^{n}$, $ x_{2} \in \mathbb Z_{g_{2}}^{n} $ and $ x_{3} \in \mathbb Z_{g_{3}}^{n}  $ be
mutually pairwise balanced and 3-qualitatively independent vectors. Let $h$ be a
positive integer so that $h\leq \min\{\lfloor\frac{n}{g_1g_2}\rfloor, \lfloor\frac{n}{g_1g_3}\rfloor\}$ and for $h\geq 3$,  $$\lfloor\frac{n}{g_1g_2}\rfloor \equiv 0  \mbox{~~~mod~~} h.$$ 
Then there exists a balanced vector $y\in \mathbb Z_{h}^{n}$ such that $\lbrace x_{1} , x_{2} , y \rbrace$ are 3-qualitatively independent 
and $y$ is pairwise balanced with each $x_i$ for $i=1,2,3$.
\end{lemma}

\begin{proof}
Construct a tripartite 3-uniform hypergraph $H$ corresponding to $x_1, x_2$ and $x_3$ as described above.
Lemma \ref{BP} implies that $H$ is a balanced hypergraph 
having maximum degree $\Delta(H)=h$. Theorem \ref{balanced} says that $E(H)$ is union of $h$ edge-disjoint matching $F_0$, $F_1, \dots, F_{h-1}$.
Identify those points of $H$ which corresponds to the same point of $G$, then $F_0$, $F_1$, \ldots, $F_{h-1}$
are mapped onto certain edge disjoint spanning  partial hypergraphs $F^{\prime}_0$, $F^{\prime}_1$, \ldots, $F^{\prime}_{h-1}$ of $G$. 
These $h$ edge-disjoint spanning partial hypergraphs $F^{\prime}_0$, $F^{\prime}_1$, \ldots, $F^{\prime}_{h-1}$ of $G$ form a partition of $E(G)=[1,n]$ which we use to build a balanced vector $y\in \mathbb Z_h^n$.
Each edge disjoint spanning partial hypergraph corresponds to a symbol in $\mathbb Z_h$ and each edge corresponds to an index from $[1,n]$.
Suppose edge disjoint spanning partial hypergraph $F^{\prime}_d$ corresponds to symbol
$d\in \mathbb Z_h$. For each edge $i$ in $F^{\prime}_d$, define $y(i)=d$.  We have
$$\lfloor\frac{n}{g_1h}\rfloor\leq d_{{F_d}^{\prime}}(P_a)\leq \lceil\frac{n}{g_1h}\rceil $$  for $d=0,1,\ldots,h-1.$
It follows from similar arguments as in Lemma \ref{lemma2}. 
 Similarly $y$ is pairwise balanced with $x_2$ and $x_3$.
Now we show that $x_1$, $x_2$, $y$ are 3-qualitatively independent. 
Let $(a,b,d)\in \mathbb Z_{g_1} \times \mathbb Z_{g_2} \times \mathbb Z_h$ be a tuple of symbols. 
For every $a\in \mathbb Z_{g_1}$, $b\in\mathbb Z_{g_2} $, there are $h$ hyperedges containing $P^l_{ab}$ and $Q^l_{ba}$ in $H$, 
and they will all appear in different matchings $F_0$, $F_1$, \ldots, $F_{h-1}$.
This ensures that each spanning partial hypergraph contains at least one $P_a-Q_b$ hyperedge for every $a\in \mathbb Z_{g_1}$, $b\in\mathbb Z_{g_2} $.
Whence there exists at least one hyperedge $i\in F^{\prime}_d$ such that 
$x_1(i)= a$, $x_2(i)=b$ and $y(i)=d$. Thus, $x_1$, $x_2$ and $y$ are 3-qualitatively independent. We need to show that $y$ is balanced.  
This corresponds to 
each matching $F_i$ contains either $\lfloor\frac{n}{h}\rfloor$ or $\lceil\frac{n}{h}\rceil$ hyperedges.   
Suppose  we have two matching $F_0$ and $F_1$  that differ by size more than 1, say $F_0$ smaller and $F_1$ larger. 
Every component of the union of $F_0$ and $F_1$ could be an alternating even cycle hypergraph or alternating path. 
 Note that it must contain a path, otherwise their sizes are equal. We can find an alternating path in the union hypergraph
that contains more edges from $F_1$ than $F_0$. Swap the $F_1$ edges with the $F_0$ edges in this alternating path. 
Then the resultant graph has $F_0$ 
increased in size by 1 hyperedge, and $F_1$ decreased in size by 1 hyperedge. Continue this process on $F_0$, $F_1$, \ldots, $F_{h-1}$  
until the sizes are correct.    
\end{proof}

 \begin{figure}[h]
 \centering
 \includegraphics[width=6cm]{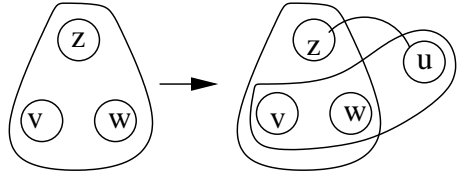}
 \caption{Single-vertex hyperedge hooking II}
 \label{fig:xyz}
\end{figure}

\begin{proposition} \label{prop2} Let $H$ be a weighted hypergraph with $k$ vertices and $H'$ be the weighted 
hypergraph obtained from $H$ by single vertex hyperedge hooking II operation with $u$ as a new vertex with $w(u) $ such that $PW(H)=PW(H^{\prime})$ and for $w(u)\geq 3$
$$\lfloor\frac{n}{w(v)w(w)}\rfloor \equiv 0  \mbox{~~~mod~~} w(u).$$ 
 Then, there exists a balanced $CA( n, H, \prod_{i=1}^{k} g_{i})$ 
 if and only if there exists a balanced $CA(n, H^{'}, w(u)\prod_{i=1}^{k} g_{i})$.
\end{proposition}
\begin{proof}
  If there exists a balanced CA($ n, H', w(u)\prod_{i=1}^{k} g_{i} $) then by deleting the row 
 corresponding to the new vertex $u$ we can obtain a $CA(n, H, \prod_{i=1}^{k} g_{i})$. 
 Conversely, let $ C^{H} $ be a balanced CA($ n, H, \prod_{i=1}^{k} g_{i} $).   If $H^{\prime}$ is obtained from $H$ by a single vertex hyperedge hooking II of a new vertex $u$ with a new hyperedge $\{u,v,w\}$ and a new edge $\{u,z\}$ where $\{ v,w,z\}$ is an existing hyperedge in $H$ and $w(u)$ such that $w(u)w(v)w(w)\leq n$ and $w(u)w(z)\leq n$. Using Lemma \ref{lemma3}, we can build a length-$n$
 vector $y$ such that $\{y,x_1,x_2\}$ is 3-qualitatively independent and $y$ is pairwise balanced with $x_1,x_2,x_3$, where $x_1,x_2,x_3$ are length-$n$ vectors correspond 
 to vertices $v,w,z$ respectively. The array $C^{H^{\prime}}$ is obtained by appending row $y$ to $C^H$. 
 \end{proof}
\begin{theorem}\label{thm5}
 Let $H$ be a weighted 3-uniform cycle $(v_0,E_1,v_2,E_2,v_3,E_3,v_0)$ of length-3 on five vertices 
 with $E_1=\{v_0,v_1,v_2\}$, $E_2=\{v_{1},v_2,v_{3}\}$ and $E_3=\{v_3,v_4,v_0\}$. Let $g_i$ denote the weight of vertex $v_i$. 
  Let $E_1$ be a hyperedge in $H$ with $g_0g_1g_2=PW(H)$.
  If $g_0\equiv 0 \mod g_3$ and    
  $g_3\leq \min \{ g_0, \max \{ g_1, g_2\} \}$
 then there exists a balanced 
 3-$CA( n, H, \prod_{i=0}^{4} g_{i} )$ with $n=PW(H)$. 
 \end{theorem}

\begin{proof} Let $H_1$ be a hypergraph with single hyperedge $E_1$. There exists a balanced  3-$CA(n,H_1, \prod_{i=0}^2 g_i)$.
Let $H_2=H_1\cup\{E_2, \{v_0,v_3\} \}$.  From Proposition \ref{prop2}, as $g_0\equiv 0 \mod g_3$ and    
  $g_3\leq \min \{ g_0, \max \{ g_1, g_2\} \}$,
 there exists  a balanced 3-$CA(n,H_2, \prod_{i=0}^3 g_i)$. Let $H_3$ be the hypergraph 
obtained from $H_2$ by replacing edge $\{v_0,v_3\}$ by hyperedge $\{v_0,v_3,v_4\}$. Note that $H_2=H$. As $g_0g_3g_4\leq n$, using single-vertex 
hyperedge hooking I operation, we get  a balanced covering array 3-$CA(n,H, \prod_{i=0}^4 g_i)$
\end{proof}

\section{Conclusions and Open Problems} In this paper, we study construction of optimal mixed covering arrays on 3-uniform hypergrahs.  This  paper extends the work done by  Meagher,  Moura, and  Zekaoui \cite{mixed} for mixed covering arrays on graph to mixed covering arrays on hypergarphs. We gave five hypergraph operations that enable us to add 
new vertices, edges and hyperedges to a hypergraph. These operations have no effect on the covering array number
of the modified hypergraph. Using these hypergraph operations, we build optimal mixed covering arrays for special classes of hypergraphs, e.g., 3-uniform $\alpha$-acyclic hypergraphs, 
3-uniform interval hypergraphs, 3-uniform conformal hypertrees, and specific 3-uniform cycles. The five basic hypergraph operations introduced here may be useful 
for obtaining optimal mixed covering arrays on other classes of hypergraphs.    It is an interesting open problem to find  optimal mixed covering arrays on  conformal hypergraphs, tight cycle hypergraphs, Steiner triple systems, etc. 
 \\

\noindent{ \bf Acknowledgement:} We are grateful to Jaikumar Radhakrishanan, Tata Institute of Fundamental Research, Mumbai,  and Sebastian Raaphorst, University of Ottawa, for useful discussions and their comments 
on the proofs of Lemma \ref{lemma1} and Lemma \ref{lemma2}. 
The first author gratefully acknowledges support from the Council of Scientific and Industrial Research (CSIR), India, during the work 
under CSIR senior research fellow scheme.

\end{document}